\documentclass{llncs}

\usepackage{amsmath}
\usepackage{amsfonts}
\usepackage{amssymb}
\usepackage{complexity}  
\usepackage{graphicx}
\usepackage{todonotes}
\usepackage{color}
\usepackage{enumitem}
\usepackage{ntheorem} 
\newtheorem*{facts}{Facts}

\usepackage{tikz}
\tikzstyle{vertex}=[circle,fill=black!0,minimum size=4pt,inner sep=0pt]

\usepackage{cleveref}

\newcommand{\netroot}[1]{r_{#1}}

\title{Locating a Tree in a  Phylogenetic Network\\ in Quadratic Time}
\author{Philippe Gambette\inst{1} \and Andreas D. M.  Gunawan\inst{2} \and Anthony Labarre\inst{1} \and  \\
  St\'ephane Vialette\inst{1} \and Louxin Zhang\inst{2}}
\authorrunning{Philippe Gambette et al.}
\institute{
Universit{\'e} Paris-Est, LIGM (UMR 8049), UPEM, CNRS, ESIEE, ENPC, F-77454, Marne-la-Vall{\'e}e, France\\
\and 
Department of Mathematics, National University of Singapore}

\begin{document}
\maketitle

\begin{abstract}
A fundamental problem in the study of phylogenetic networks is to determine whether or not a given phylogenetic network 
contains a given 
phylogenetic tree.     We  develop a quadratic-time algorithm for this problem for binary nearly-stable phylogenetic networks.
We also show that the number of reticulations in a reticulation visible or nearly stable phylogenetic network  is bounded from above  by 
a function linear in the number of taxa. 
\end{abstract}

\section{Introduction}

Genetic material can be transferred between organisms by hybridization, recombination and horizontal gene transfer besides traditional reproduction. Recent studies in comparative  genomics suggest that these ``lateral''  processes are a driving force in evolution which shapes the genome of a species~\cite{Chan_13_PNAS,McBreen_06_TrendsPS,Treangen_11_PLOSGenentics}.  
Accordingly, phylogenetic networks have commonly been used to model reticulate evolutionary histories of species~\cite{Chan_13_PNAS,Dagan_08_PNAS,Marcussen_SysBiol_12}. A plethora  of methods for reconstructing  reticulate  evolutionary histories of species 
and related algorithmic issues have  extensively been  studied over the past two decades~\cite{Gusfield_14_Book,Huson_Book,Moret_04_TCBB,Nakhleh_13_TREE,Parida_10,Wang_01_JCB}.

A \emph{phylogenetic network} is an acyclic digraph with a set $X$ of  labeled leaves (that is,  vertices of outdegree zero) and a  root (having indegree zero). 
The leaves  
are in one-to-one correspondence with 
a collection  of taxa  under study, whereas the unique root represents their least common ancestor.  Vertices with indegree one   represent speciation events. Vertices of indegree at least two represent an evolutionary process by which genetic material was horizontally transferred from one species to another.

A fundamental question in the study of phylogenetic networks is  to determine whether a tree is \emph{displayed} by a phylogenetic network over the same set of taxa (in a sense we  define precisely below). This problem is called the \emph{tree containment problem}~\cite{Huson_Book}. Answering this question is indeed useful to validate and justify a  phylogenetic network model  by testing whether it displays existing phylogenies over a set of taxa under study. 

The problem  is \NP-complete in general~\cite{Kanj_08_TCS},
 even on the more restricted class of \emph{tree-sibling time-consistent regular networks}~\cite{van_Iersel_2010_IPL}. Although great effort has been devoted to the study of that problem,  it has been shown to be polynomial-time solvable only for a couple of interesting classes of phylogenetic networks,  namely, \emph{normal} networks and \emph{tree-child} networks~\cite{van_Iersel_2010_IPL}. 
Determining the complexity of the tree containment problem for a class of phylogenetic networks that properly contains 
tree-child 
networks,  particularly  those with the so-called \emph{reticulation-visibility property}, is an open problem~\cite{Huson_Book,van_Iersel_2010_IPL}. 

In this paper, we study the tree containment problem for \emph{nearly stable} phylogenetic networks (defined in the next section), which 
generalize  
normal and 
tree-child 
networks. Recombination histories of viruses, hybridization histories of plants, and histories of horizontal gene transfers reported in literature often 
satisfy the property that defines those networks~\cite{Jenkins_PlosOne_12,Marcussen_SysBiol_12}.   Our key results include: (i) the number of reticulations in a reticulation-visible or nearly stable phylogenetic network  is linearly bounded from above in terms of the number of taxa; and (ii) the tree containment problem for nearly stable phylogenetic networks can be solved in  quadratic time. Omitted proofs and details will appear in the extended version.

\section{Concepts and Notions}


A {\it (phylogenetic) network} on a set $X$ of taxa 
is a directed acyclic graph with a single root (a vertex with indegree 0) which satisfies the following properties: (i)  its leaves (vertices with outdegree 0) are in one-to-one correspondence with the taxa in $X$; (ii) there are no vertices with both indegree one and outdegree one; and (iii) there is a path from the root to any other vertex.  We identify each leaf with the taxon corresponding to it and refer to the directed edges (tail, head) as \emph{branches}.

In a network, 
\emph{reticulation vertices} (or simply \emph{reticulations}) are vertices  with indegree at least two and outdegree one; {\it tree vertices} are vertices with indegree one and outdegree at least two. A branch is a {\it tree branch} if it ends at a tree vertex; it is called a {\it reticulation branch} otherwise. 

A network  is {\it binary}  if its root, leaves and the other vertices have degree 2, 1 and 3, respectively.   
A \emph{phylogenetic tree} is simply  a binary  network without reticulations.

For a binary  network $N$, we shall use $\netroot{N}$ to denote the root of  $N$. 
Let $x$ and $y$ be vertices in $N$. We say that $x$ is a {\it parent}  of $y$ and $y$ is a {\it child}  of $x$ if $(x, y)$ is a branch.
More generally, we say that  $x$ is an {\it ancestor} of $y$ and equivalently $y$ is a {\it descendant} of $x$ if there is a directed path from $x$ to $y$.
 A vertex $x$ in $N$ is a \emph{stable ancestor} of a vertex $v$
if it belongs to all directed paths from $\netroot{N}$ to $v$.
We say that $x$ is \emph{stable} if there exists a leaf $\ell$
such that $x$ is a stable ancestor of $\ell$.

\begin{proposition} 
\label{three_facts}
Let $N$ be a binary network.  The following facts hold.\vspace{-0.5em}
\begin{itemize}
\item[{\rm (1)}] 
A vertex  is stable if it has a stable tree child.
 \item[{\rm (2)}] 
A reticulation is stable if and only if  its unique child is a stable tree vertex.
\item[{\rm (3)}] 
If a tree vertex is stable, then its children cannot both be reticulations. 
\end{itemize}
\end{proposition} 

A network is a {\it tree-child} network if every vertex has a  child that is a tree vertex~\cite{Cardona_09_TCBB}. It can be proved that a network is a tree-child network if and only if every vertex is stable.  
It is \emph{reticulation-visible} 
if all its reticulations are stable~\cite{Huson_Book}.  It is {\it nearly stable} if for every vertex, either that vertex is stable or its parents are.  

\emph{Contracting} a branch $(u, v)$ means replacing it with a single vertex $w$ in such a way that all neighbors of $u$ and $v$ become neighbors of $w$.  Given a binary phylogenetic tree $T$ and a binary network $N$, we say that $N$ \emph{displays} $T$ if there is a spanning subtree $T'$ of $N$ that is a \emph{subdivision} of $T$,
i.e. $T'$ has the same vertex set as $N$ and $T$ can be obtained from $T'$ 
by contracting all branches in $T'$ incident with the vertices with outdegree 1 and indegree 1, all branches incident with the ``dummy leaves'' (leaves in $T'$ that correspond to tree vertices in $N$), and all branches incident with a vertex of indegree 0 and outdegree 1.
\Cref{fig:example-tc}  shows an example of  a phylogenetic network $N$ and a tree that is displayed in $N$.

\begin{figure}[htbp]
\centering

	\includegraphics[width=0.7\textwidth]{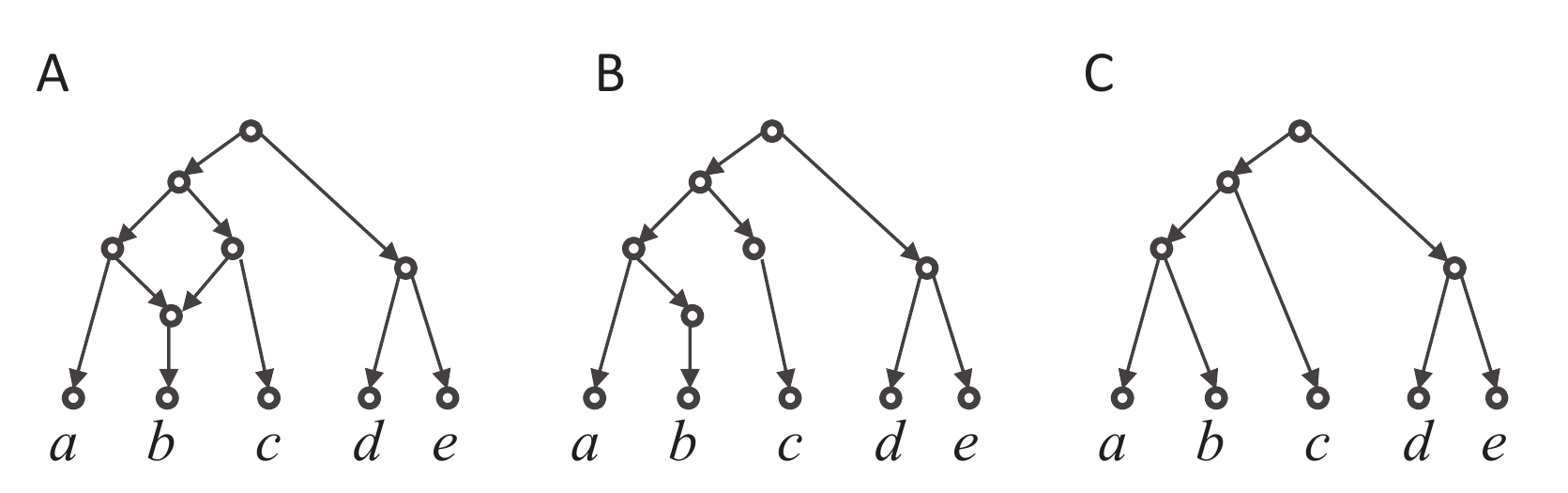}
\caption{ (A) A phylogenetic network. (B) A spanning subtree of $N$ obtained after the reticulation branch between the parents of $c$ and $b$ is removed.  (C) A tree displayed in $N$ through the subtree in (B).}
\label{fig:example-tc}
\end{figure}

 In this work, we study the {\it tree containment problem} (TCP), which is that of determining whether a phylogenetic tree is displayed by a network or not.   

\section{How Many Reticulations in a Network?}

 An arbitrary network with $n$ leaves can have a very large number of reticulations.  To analyze the time complexity of an algorithm 
designed for solving a network problem, we need to bound the size of the network by a function of $n$.  

Removing a reticulation branch from each reticulation in a binary network $N$
yields a spanning subtree $T'$.  All leaves in $N$ are 
still  leaves in $T'$, but  $T'$ may 
additionally 
contain some ``dummy leaves"  that correspond to tree vertices whose outgoing branches have both been removed.  The following lemma says that 
it is always possible to remove proper reticulation branches so as to obtain a tree without dummy leaves.

\begin{lemma}
\label{lemma_no_dead_end}
 Let $N$ be a binary reticulation-visible phylogenetic network. 
We can determine  which reticulation branch to remove at each reticulation  so that the  tree obtained after 
removing 
the selected branches 
contains no 
dummy leaves.
\end{lemma}
\begin{proof}  
Let $T$ be a tree obtained from $N$ 
by removing exactly one reticulation branch incident to each reticulation.
In order for $T$ not to contain any dummy leaves, we need to guarantee that
the reticulation branches to be removed are incident with different tree vertices. In other words, the branches to be removed form a matching that covers every reticulation  in $N$. 
Since $N$ has the reticulation-visibility property,  the  parents of each reticulation  are both  tree vertices (\Cref{three_facts}). 
Such a set of reticulation branches exists and can be found by applying Hall's Theorem to a bipartite graph 
with tree vertices and reticulations as vertex sets and reticulation branches as edges. 
Since each reticulation  is the head of two reticulation branches and each tree vertex is the tail of at most two reticulation branches, there exists a matching that covers all the reticulations (see a result of N. Alon on page 429 in \cite{Bondy_Book}).  \qed
\end{proof}

\begin{figure}	
\begin{center}
	\includegraphics[width=0.9\textwidth]{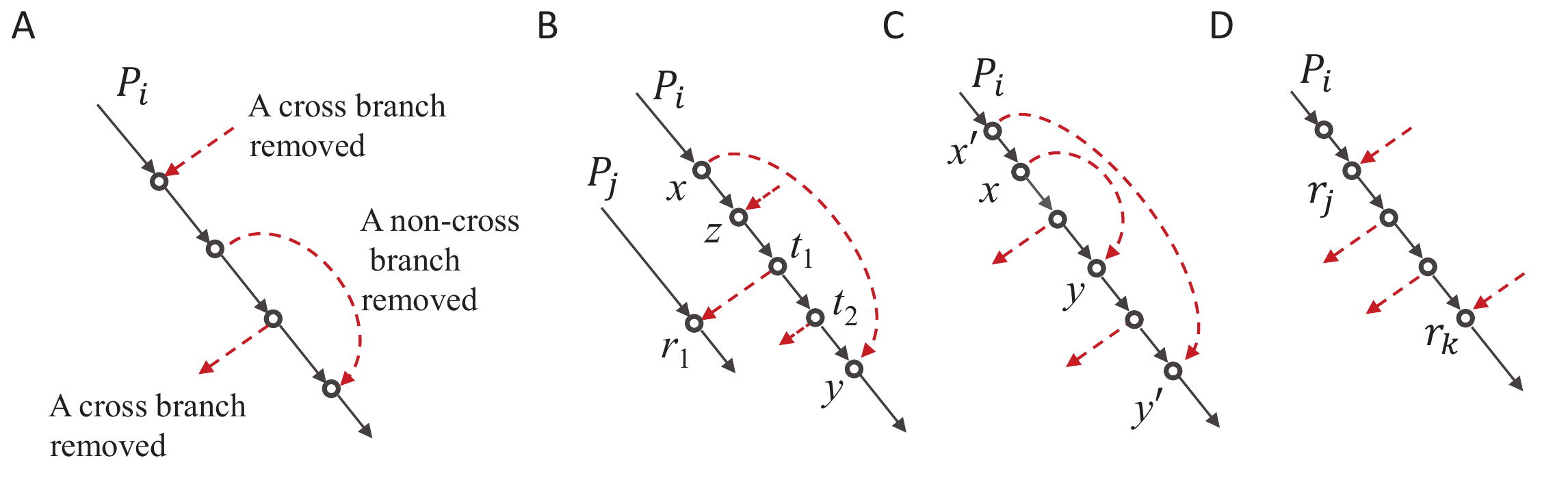}
\end{center}
\caption{Illustration of the different cases in the proof of \Cref{linear_bound}.
{\bf A}. Definition of cross and non-cross branches removed from a path. 
{\bf B}. The branch $(x, y)$ is a non-cross branch removed from a path. 
Assume that a cross branch $(z', z)$ has been removed from a reticulation $z$ inside the segment from $x$ and $y$, where $z'$ is not shown,  
and two cross branches have also been removed from two tree vertices $t_1$ and $t_2$  between $z$ and $y$. 
{\bf C}.  Some cross branches must have been removed from their tails located between the heads of  two non-cross branches  that are removed from a path (in this case, between $y$ and $y'$).
{\bf D}. If  two cross branches have been removed from two reticulations in a path, then the upper reticulation ($r_j$ here)  is not stable.
\label{linear_bound_figure}}
\end{figure}

\begin{theorem}
\label{linear_bound}
Let $N$ be a binary reticulation-visible phylogenetic network with $n$ leaves. Then $N$ has at most $4(n-1)$ reticulations. 
\end{theorem}
\begin{proof} 
Assume $N$ contains $m$ reticulations.
By \Cref{lemma_no_dead_end}, we can obtain a tree $T$ without dummy leaves by removing $m$
reticulation branches from $N$.  Since $N$ is  binary, an internal vertex in $T$ has either one or 
two children; equivalently, $T$ is a subdivision of a rooted binary tree $T'$ over the same leaves as $N$.  
Therefore, $T'$ has $n-1$ internal vertices (including its root) of outdegree 2 and there are $2n-2$ paths $P_i$ ($1\leq i\leq 2n-2$)
satisfying (i) the ends of each $P_i$ are either the root of $T$, a leaf, or internal vertices of outdegree 2,
and (ii) each internal
vertex  of  $P_i$ has both indegree and outdegree 1 if $P_i$ consists of two or more branches. 

 For each path $P_i$ of  length $\geq 2$, an internal vertex 
of $P_i$ is either a tree vertex of $N$, whose outgoing branch not in $P_i$  has been removed,  or a reticulation, whose incoming branch not in $P_i$ has been removed. For convenience of discussion, we 
 divide the removed reticulation branches into {\bf cross} and {\bf non-cross}  branches (with respect to $T$)
(\Cref{linear_bound_figure}A).  A removed branch is called a  \emph{cross branch} if its tail and head are located on two different paths $P_i$ and $P_j$, $i\neq j$,
otherwise it's called a \emph{non-cross branch}.   We first have the following facts.

\parbox{4.6in}{
\begin{facts}\label{thmfacts}
\ \vspace{-2em}\\
\begin{enumerate}[label=(\arabic*)]
\item[{\rm (1)}] 
If  $(x, y)$  is a non-cross branch removed from $P_i$, 
then at least one cross branch has been removed from its tree vertex tail in the segment $P_i[x, y]$ from $x$ to $y$ of $P_i$, 
and there is no reticulation in $P_i[x, y]$ other than $y$.

\item[{\rm (2)}] Let $(x, y)$ and $(x', y')$ be two non-cross branches removed from  $P_i$, where $y$ is an ancestor of $y'$. Then there exists at  least one cross branch being removed from its tree vertex tail located between $y$ and $y'$ {(\rm \Cref{linear_bound_figure}C)}.

\item[{\rm (3)}] There are at least as many cross reticulation branches removed as non-cross reticulation branches. 
\end{enumerate}
%
%
\end{facts}
\begin{proof} 
(1)  
Since $N$ contains no parallel branches, 
$P_i[x, y]$ has at least three vertices,  
so it suffices to prove that  $y$ is the only reticulation in $P_i[x, y]$.  

Assume on the contrary that 
 a branch $(z', z)$ has been removed from a reticulation $z$ in $P_i[x, y]$ (\Cref{linear_bound_figure}B).
Then there is a
path including $(x, y)$ from 
$\netroot{N}$ 
to a leaf below $y$
that avoids 
$z$, so $z$ is not stable on any leaf below $y$ (and hence below $z$) 
 in $T$ (and hence in $N$). %
Moreover, since $T$ is a subtree of $N$, $z$ cannot  be stable in $N$ on any leaf that is not  below $z$ in $T$. $N$ and $T$ have  the same leaf set, hence
 $z$ is not stable in $N$, contradicting  the reticulation-visibility property.

(2) Note that $y$ and $y'$ are  reticulations in $N$. By Fact (1) above, $y$ must be above $x'$, and there is a cross brach removed from its tree vertex tail located between $x'$ and $y'$.

(3)  By Facts (1) and (2), we can establish an injective map from the set of non-cross reticulation branches to that of cross ones. 
Hence, the statement in this part is also true.   \qed
\end{proof}
}

\noindent Assume  at least $2n-1$  cross branches $(t_i, r_i)$ have been removed from the $2n-2$ paths 
$P_i$.  At least two heads $r_j$ and $r_k$ are on the same path $P_i$  (\Cref{linear_bound_figure}D). Using an argument 
similar to that used in the proof of { Fact} (2),  one of  $r_j$ and $r_k$ which is upstream in $P_i$ is not stable,  a contradiction.   Therefore, at most $2n-2$ cross branches have been removed 
to produce $T$. By { Fact} (3), there are also at most $2n-2$ non-cross branches removed during the process.  Since we removed 
one incoming branch for each reticulation, we conclude that there are at most $4(n-1)$ reticulations in $N$.
 \qed
 \end{proof}

\begin{lemma}
\label{lemma:bin-NS-to-bin-RV}
 Let $N$ be a binary  nearly stable 
 network, and let $U_{\rm ret} (N)$ (resp. $S_{\rm ret}(N)$) denote the number of all unstable (resp. stable) reticulations in $N$.  We can transform $N$ into a binary reticulation-visible 
 network $N'$ with the property that  $N'$ has the same leaf set as $N$ and  $S_{ret}(N) \leq S_{ret}(N') \leq S_{ret}(N) + U_{ret}(N)$.
\end{lemma}
\begin{proof}
Let $a$ be an unstable reticulation in $N$, whose child is denoted by $b$. Since $N$ is nearly stable, 
$b$ is stable.  By \Cref{three_facts}(2),  $b$ is a stable reticulation.
 Let $c$ denote a parent of $a$; 
then $c$ is stable by definition of $N$, 
and it is a tree vertex by \Cref{three_facts}(2).
Let $d$ denote the other child of $c$. 
Since $c$ is stable, $d$ is a tree vertex (\Cref{three_facts}(3)). 
In addition,  $d$ is stable.

Assume on the contrary that $d$ is unstable. Then  both its children must be stable by the nearly-stable property of $N$. Hence,  by \Cref{three_facts}(2)  and the fact that $d$ is unstable, both its children are stable reticulations.  Since $a$ is unstable, $a$ is not a child of $d$.
This implies that $c$ is unstable, a contradiction.

Finally, let $e$ be the parent of $c$.  
$f$ be the other parent of $a$ 
and 
$g$ be the other parent of $b$ (see \Cref{Figure1}).
 Note that $g \neq f$.  Otherwise,  $f$ is unstable, contradicting that there are no two consecutive unstable vertices. 
To transform $N$ into a binary reticulation-visible 
network,  we remove unstable vertex $a$ by first removing the branch $(c, a)$, and then contracting  the paths $f$-$a$-$b$ and $e$-$c$-$d$ into branches $(f, b)$ and $(e, d)$. 
Both $b$ and $d$ are clearly still stable in the resulting network.  By rewiring around every unstable reticulation in $N$, we produce a binary reticulation-visible network $N'$. The inequality 
follows 
from the fact that no stable reticulation is removed, and no new reticulation is created during the rewiring.  \qed
\end{proof}

\begin{figure}	
\begin{center}
	\includegraphics[width=0.8\textwidth]{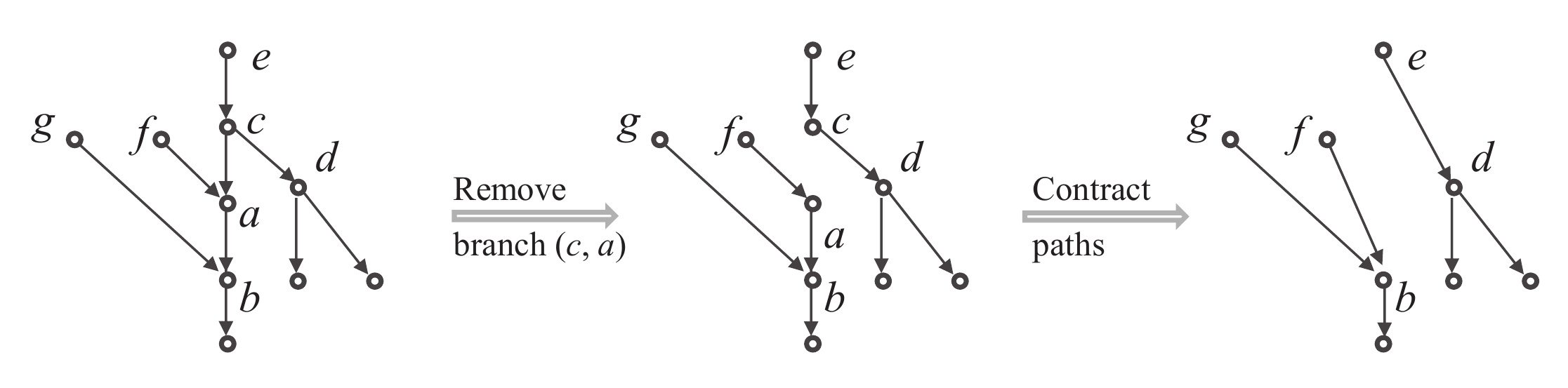}
\end{center}
\caption{ {\bf A}  An unstable reticulation $a$, its  stable child $b$ and its stable parents ($c$ and $f$) in the original network  $N$. To transform $N$ into a reticulation-visible network, 
we remove the incoming reticulation branch $(c, a)$ ({\bf B})  and then  contract paths $e$-$c$-$d$ and $f$-$a$-$b$ ({\bf C}). The rewiring eliminates the unstable reticulation vertex $a$.
\label{Figure1}}
\end{figure}

\begin{lemma}
\label{lemma:U-is-at-most-two-S} For a binary nearly stable 
network $N$, $U_{ret}(N) \leq 2 S_{ret}(N)$.
\end{lemma}
\begin{proof} 
Directly follows 
from the fact that an unstable reticulation must have a stable reticulation as its child, and any stable reticulation can be the child of at most two unstable reticulations.
\qed
\end{proof}

\begin{theorem} 
\label{Bound_Thm}
Let $N$ be a binary nearly stable 
network with $n$ leaves. Let $T(N)$ denotes the number of tree vertices in $N$. Then:
\begin{enumerate}[label=(\roman*)]
\item $N$ has at most $12(n-1)$ reticulations;

\item $|T(N)|\leq 13(n-1)$ and $|E(N)|\leq 38(n-1)$.
\end{enumerate}
\end{theorem}

\begin{proof} 
(i)  
\Cref{linear_bound} and \Cref{lemma:bin-NS-to-bin-RV,lemma:U-is-at-most-two-S} 
imply $S_{ret}(N) + U_{ret}(N) \leq 3 S_{ret}(N) \leq 3 S_{ret}(N') \leq 3(4n-4)=12(n-1)$.

\noindent (ii)   We can think of the network as a flow network, 
 with  $\netroot{N}$ as source and the $n$ leaves as  sinks. Hence,  the number of tree vertices equals $n-1$ plus the number of reticulations, that is,  at most $13(n-1)$ (by (i)). Since the outdegree of the root is two, and the outdegrees of each tree and reticulation  vertex are  2 and 1, respectively,  $N$ has $2(13n-13)+12(n-1)=38(n-1)$ branches at most. \qed
 \end{proof}
\section{A Quadratic-Time Algorithm for the TCP}

In this section, we shall present a quadratic-time algorithm for solving the TCP.  If a given 
network  $N$ and a given reference 
tree $T$ contain a common subphylogeny, then we can simplify the task of determining whether $N$ displays $T$ by 
replacing the common subphylogeny by a new leaf.  Therefore, without loss of generality, we assume that $N$  does not contain a subphylogeny with two or more leaves. We call this property  the {\it subphylogeny-free property}.

\begin{figure}[!th]	
\begin{center}
	\includegraphics[width=0.9\textwidth]{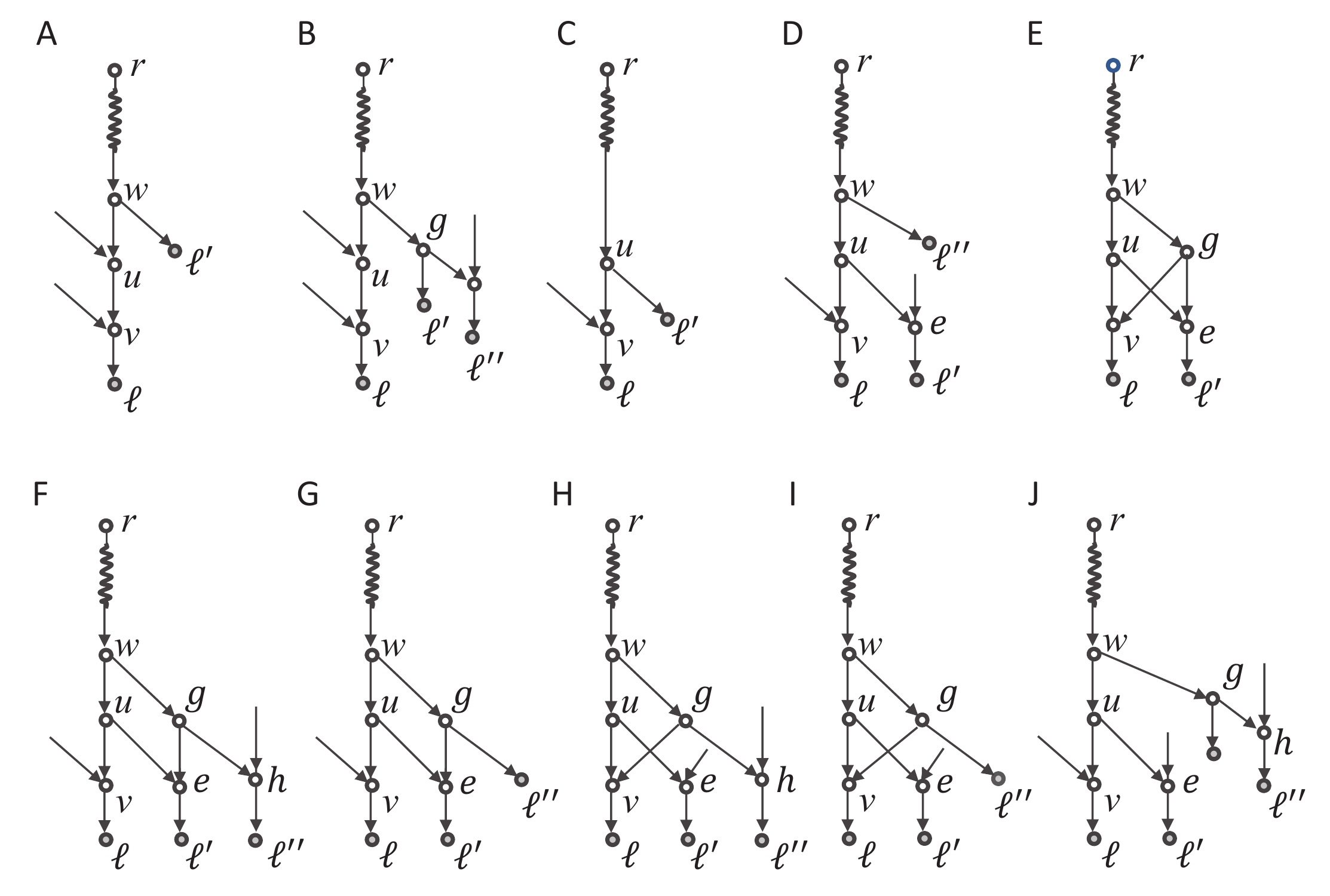} 
\end{center}
\caption{ All ten possible subnetworks at the end  of a longest path in a  nearly stable network. Here, $r$ is the network root and the directed path from $r$ to  $w$ is represented by a coiled path. The parent $w$ of $u$ is not shown in {\bf C}.}
\label{5_cases}
\end{figure}

\begin{lemma}
\label{lemma_tcp}
 Let $N$ be a nearly stable phylogenetic network satisfying the  subphylogeny-free property. Let 
$P=( r, \ldots, w,  u, v, \ell)$
be a longest root-to-leaf path of four or more vertices in $N$, where $r=r_N$ and $\ell$  the leaf end.  Then 
the subnetwork consisting of the descendants of $w$
exhibits one of the structures given in \Cref{5_cases}.
\end{lemma}
\begin{proof} 
Note that $v$ cannot be a tree vertex: since $P$ is a longest root-to-leaf path, the other child of $v$ would otherwise be a leaf, thereby contradicting our assumption that $N$ satisfies the  subphylogeny-free property. Therefore, $v$ is a reticulation. There are  two possible cases for $u$.

\begin{enumerate}
\item 
{\bf The $u$ is a reticulation:} 
Then $u$ is unstable, and $w$  must be a stable tree vertex 
(see \Cref{three_facts}(2) for both claims), 
which is stable on  $\ell$ or some other leaf.  Let $g$ be the other child of $w$. By  \Cref{three_facts}(3),
$g$ is either a tree vertex or a leaf.  If $g$ is a leaf, we obtain the subnetwork in \Cref{5_cases}A. 
If $g$ is a tree vertex, then neither of its children is a tree vertex: since $P$ is a longest path, a tree vertex child of $g$ would have two leaves as children, thereby contradicting the subphylogeny-free property. Note that $g$'s children cannot both be reticulations either, since otherwise $w$ would be unstable. Therefore, 
one child of $g$ is a leaf and the other is a reticulation with a leaf child (again because $P$ is a longest path),  
as shown in \Cref{5_cases}B.

\item {\bf The $u$ is a tree vertex:} Let $e$ denote  the other child of $u$. 
Note that $e$ cannot be a tree vertex, otherwise both its children would be leaves (since $P$ is a longest path), which would contradict our assumption that $N$ has the subphylogeny-free property. 
 If $e$ is a leaf, we obtain the subnetwork shown in  \Cref{5_cases}C. 
If  $e$ is a reticulation,  then its only child is  a leaf (again because $P$ is a longest path), so $e$ is stable on that leaf and $u$ is therefore unstable.  Since $N$ is nearly stable,  $w$ must be  a  stable tree vertex. We consider the other child $g$ of $w$ in the following subcases.

\begin{enumerate}[label=(2.\arabic*),leftmargin=.35in]
\item If $g$ is a leaf, then we have the subnetwork given in \Cref{5_cases}D.

\item If $g$ is a tree vertex  and  also a parent of $e$ and $v$, then we obtain the subnetwork in  \Cref{5_cases}E. 

\item If $g$ is a tree vertex and in addition, $g$ is a parent of $e$, but not a parent of  $v$: then $w$ is  stable on $\ell'$, the unique child of $e$.  Let $h$ be the other child of $g$; 
then $h$ cannot be a tree vertex,  since both its children would then be leaves, which would contradict our assumption that $N$ has the subphylogeny-free property.
 If $h$ is a reticulation, its child must be a leaf, since $P$ is a longest path. Thus,   we have the subnetwork given in \Cref{5_cases}F. 
If $h$ is a leaf,  we obtain the 
subnetwork in  \Cref{5_cases}G.

\item If $g$ is a tree vertex and in addition, $g$ is a parent of $v$, but not a parent of  $e$, then a discussion similar to that  of case (2.2) characterises the only two possible subnetworks (\Cref{5_cases}H and \ref{5_cases}I) in this case.

\item If  $g$ is a tree vertex and in addition, $g$ is neither  a parent of  $v$ nor a parent of $e$: then again we look at $g$'s children. Both cannot be reticulations, otherwise $w$ is unstable, a contradiction.   If neither of them is a reticulation,  then there is a subtree below $g$; 
if one of them  is a reticulation and the other is a tree vertex,  then again there is a subtree. The only possible case that remains, shown in \Cref{5_cases}J, is the case where one child is a reticulation and the other is a leaf.

\item If $g$ is a reticulation: Then $w$  unstable. This is impossible, as $w$ is a stable tree vertex. \qed
\end{enumerate}
\end{enumerate}
\end{proof}


The  subnetwork below $g$ of the structures shown in \Cref{5_cases}B, \ref{5_cases}G, \ref{5_cases}I, \ref{5_cases}J  and that below $u$ in \Cref{5_cases}C match  the following pattern:
\begin{center}
     \includegraphics[scale=0.7]{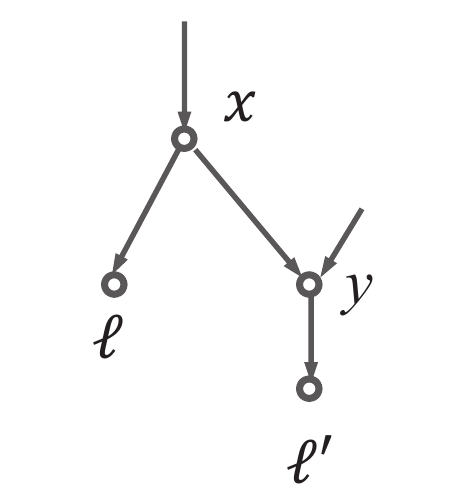} 
\end{center}
in which a leaf $\ell$  has a reticulation sibling $y$   and  a leaf nephew, $\ell'$. Such a pattern is called an {\it uncle-nephew structure}.
Note that  if $\ell$ and $\ell'$ are not siblings in a tree  displayed by $N$, then the reticulation branch $(x, y)$ 
should not be used. If $\ell$ and $\ell'$ are siblings, either $(x, y)$ or the  other  branch entering $y$  can be used.
Here, since the other branch enters $y$ from an unspecified vertex, it is simply called a \emph{dangling branch}.
It is not hard to see that  for a tree $T$ in which  $\ell$ and $\ell'$ are siblings, if $T$ is displayed in the network resulting from the  removal of  $(x, y)$, it  is also displayed in
the one after  the dangling branch  is removed.   Hence, to determine whether $N$ displays a tree $T$, we can simplify the network by eliminating $y$ using the following process:
\begin{quote}
   {\bf Uncle-Nephew Reduction}  In an uncle-nephew substructure shown above, remove the dangling branch   if $\ell$ and $\ell'$ are siblings in $T$, or remove $(x, y)$ otherwise. Then contract vertices with indegree and outdegree 1. 
\end{quote}

In each of the other cases, we can also simplify the network by using information on the input tree. 
 To summarize  how to simplify the network,  we use the following notation for each vertex $w$ in a network $N$:
    \begin{itemize}
      \item  $R(w)$ denotes the subnetwork consisting of all the descendants of $w$;
      \item $ (-, x)$ denotes  the dangling branch entering  $x$ from its parent not in $R(w)$ for $x$ in $R(w)$;
      \item $N'+(x, y)$ denotes the subnetwork obtained by adding $(x, y)$ into $N'$ for a subnetwork $N'$ of $N$ and a branch $(x, y)$ of $N$;
  \item $N'-(x, y)$ denotes the subnetwork obtained by removing $(x, y)$ from $N'$ for a subnetwork $N'$ of $N$;
  \item $p_{T}(x)$ denotes the parent of 
  a vertex $x$ in  a tree $T$.
    \end{itemize}

\begin{theorem}
\label{validation}
Let $N$ be a binary nearly stable network with no uncle-nephew structure, and $T$ a tree with the same set of labeled leaves. Let $w$ be a tree vertex in $N$.
Define $N'$ as follows. 
\begin{enumerate}
\item[{\rm (i)}] When $R(w)$ 
matches the structure of 
\Cref{5_cases}A,  define $N'=N-(w, u)$ if $\ell$ and $\ell'$ are not 
sibling in $T$ and $N'=N-\{(-, u), (-, v)\}$ otherwise. 
\item[{\rm (ii)}]  When $R(w)$
matches the structure of  \Cref{5_cases}D, define 
$N'=N-(-, v)$ when $\ell$ and $\ell''$  are siblings, or when $\ell$ and $\ell'$ are  siblings and their parent is a sibling of $\ell''$  in $T$, and 
 $N'=N-(u, v)$ otherwise.
\item[{\rm (iii)}] When $R(w)$ 
matches the structure of 
\Cref{5_cases}E,
define $N'=N-\{(u, e), (g, v)\}$.
\item[{\rm (iv)}] When $R(w)$ 
matches the structure of 
\Cref{5_cases}F,
                   define $N'=N-\{ (g, e), (-, v)\}$ if $\ell$ and $\ell'$ are siblings in $T$ and 
                  $N'=N-(u, e)$ otherwise. 
\item[{\rm (v)}] When $R(w)$ 
matches the structure of 
\Cref{5_cases}H,
           define $N'=N-\{(g, v),(-,e)\}$ if $\ell$ and $\ell'$ are siblings in $T$ and $N'=N-(u, v)$ otherwise.
\end{enumerate}
Then $N'$ is nearly stable and  $N$ displays $T$  only if $N'$ displays $T$.
\end{theorem}
\begin{proof}
Since none of the simplifications removes any leaf and all of them only reduce possible paths from $\netroot{N}$ to a leaf, the resulting network $N'$ is nearly stable.

 Assume $R(w)$ is the subnetwork in \Cref{5_cases}A and $N$ displays $T$.  Then there exists a subtree $T'$ of $N$ that is a subdivision of $T$ and let  $p_T(\ell)$ corresponds  $x$ in $T'$. 
Clearly, $x$ is of degree 3 and hence a tree vertex in $N$. 
 We consider two cases.

CASE A. Leaves $\ell$ and $\ell'$ are not siblings in $T$.

We first have that $x\neq u$,  $x\neq v$ for $u$ and $v$  in  \Cref{5_cases}A.   We also have $x\neq w$. Otherwise, 
$\ell'$ must be a child of $x$ in $T'$ and $\ell$ is a sibling of  $\ell'$  in $T$, a contradiction. 
Therefore,  the path from $x$ to $\ell$ contains two or more vertices and $v$ is the parent of $\ell$ in this path.
If  $u$ is the parent of $v$ in the same path, neither  $(-, v)$ nor $(w, u)$ is in $T'$, 
indicating that $N'= N-(w, u)$  also displays $T$.  

If $p_{T'}(v)\neq u$ in the same path,  then $(u, v)$ is not in $T'$ and hence $u$ becomes a dummy leaf in $T'$,  as there is no leaf other than $\ell$ below $u$ in $R(w)$. If $(w, u)$ is in $T'$,  then $(-, u)$ is not in $T'$ and 
$T'+(-, u)-(w, u)$ is a subtree of $N'$ in which only the dummy leaf $u$ is relocated.  Hence, $N'$ also displays $T$.

CASE B.  Leaves $\ell$ and $\ell'$ are siblings in $T$. 

Then $x$ is a common ancestor of $\ell$ and $\ell'$ in $N$. If $x=w$,  the path from $x$ to $\ell$ in $T'$ must be $w, u, v$, 
as this is only path from $w$ to $\ell$ in $N$.  Hence, $(-, u)$ and $(-, v)$ are not in $T'$. Therefore,  $T'$ is a subtree of $N'$ and  $N'$ also displays $T$.  

If $x\neq w$, then $x$ is an ancestor of $w$ and hence $w$ is the parent of $\ell'$ in the path from $x$ to $\ell'$ in $T'$.
  Note that $p_{T'}(\ell)=v$. If  $p_{T'}(v)=u$, then 
 $(-, u)$ is in $T'$, but 
 both $(-, v)$ and $(w, u)$ are not.    $T''=T'+(w, u)-(-, u)$ is  a subtree  of $N'$.
Noting  that $T''$ is also  a  subdivision of $T$,    $N'$ displays $T$.

If  $p_{T'}(v)\neq u$,  then $(-,  v)$ is in the path from $x$ to $\ell$ in $T'$. This implies that 
$(u, v)$ is not in $T'$ and $u$ is a dead-end in $T'$. If $(w, u)$ is in $T'$,  the subtree  $T''=T'+(u, v)-(-, v)$ of $N'$  is  a subdivision of $T$.  If $(w, u)$ is not in $T'$,  the subtree  $T''=T'+(w, u)-(-, u)-(-, v)$ of $N'$  is a subtree of $N'$. Hence, $N'$ displays $T'$. 

Similarly, we can prove that $N$ displays $T$ only if $N'$ displays $T$ when $R(w)$ is the subnetwork in 
the panels D, F, and H in \Cref{5_cases}. Note also that the subnetworks in the panels $F$ and $H$ are essentially identical (if the positions of $v$ and $e$ are switched).   Due to the limited space, the details are omitted here.
The case when $R(w)$ is the subnetwork in \Cref{5_cases}E is trivial, as 
deletion of which two reticulation branches from $v$ and $e$ does not affect outcome. 
\qed
\end{proof}

By \Cref{validation}, we are able to  determine whether a nearly stable phylogenetic network $N$ displays a binary tree  $T$ 
or not by repeatedly executing the following  tasks in turn until the resulting network $N'$ becomes a tree:
  \begin{itemize}
    \item Compute a longest path $P$ in $N'=N$; 
    \item Simplify $N'$  by considering the subnetwork at the end of $P$ according to the cases in \Cref{lemma_tcp}; 
    \item Contract degenerated reticulations in $N'$ and replace the parent of a pair of leaves appearing in both $N'$ and $T$ with a new leaf. 
\end{itemize}
 and then check if $N'$ is identical to $T$. 

Finally, we analyze the time complexity. Let $N$ and $T$ have $n$ leaves. By \Cref{Bound_Thm}, there are $O(n)$ vertices and 
$O(n)$ branches in $N$. Since we eliminate at least a  reticulation in each loop step, the algorithm stops after $O(n)$ loop steps.  In each loop step, a longest path can be computed in $O(n)$ time (\cite{Sedgewick_book}, page 661), as $N$ is acyclic; both the second and third tasks can be done in 
constant time. In summary,  our algorithm has quadratic time complexity. 
 
\section{Conclusion}

We have developed a quadratic-time algorithm for the TCP for binary nearly stable phylogenetic networks.  Our algorithm not only is  applicable  to a  superclass of tree-child networks, but also has a lower time complexity than the algorithm reported in \cite{van_Iersel_2010_IPL}.  Although  phylogenetic network models built in the study of viral and plant evolution are often nearly stable, it is interesting to know whether the TCP is polynomial time solvable or not for networks with other weak properties.

In particular, the problem remains open for binary networks with the
visibility property, but the upper bound we have presented  on the number of 
reticulation vertices of such networks, as well as our algorithm
for nearly stable phylogenetic networks,   provide definitely valuable ideas to
solve the problem, exactly or heuristically, on phylogenetic
networks with the reticulation visibility property.

\section{Acknowledgments}

The project was financially supported by Merlion Programme 2013.


\end{document}